\documentclass[conference]{IEEEtran}
\IEEEoverridecommandlockouts
\usepackage[utf8]{inputenc}
\usepackage[margin=2cm]{geometry}
\usepackage{blindtext}
\usepackage{setspace}
\usepackage{graphicx}
\usepackage{notoccite} 
\usepackage{lscape} 
\usepackage{caption} 
\usepackage{cite}
\usepackage{amsmath,amssymb,amsfonts}
\usepackage{amsthm}
\usepackage{algorithmic}
\usepackage{graphicx}
\usepackage{textcomp}
\usepackage{xcolor}
\usepackage{cancel}
\usepackage{tikz}
\usetikzlibrary{automata}
\usetikzlibrary{shapes,arrows}
\usepackage{bm} 
\usepackage{pifont}
\newcommand{\begit}{\begin{itemize}}
	\newcommand{\eit}{\end{itemize}}
\newcommand{\bseq}{\begin{subequations}}
	\newcommand{\eseq}{\end{subequations}}
\newcommand{\bpat}{\begin{pmatrix}}
	\newcommand{\epat}{\end{pmatrix}}
\newcommand{\bmat}{\begin{bmatrix}}
	\newcommand{\emat}{\end{bmatrix}}
\newcommand{\beq}{\begin{equation}}
\newcommand{\eeq}{\end{equation}}
\newcommand{\bc}{\begin{cases}}
\newcommand{\ec}{\end{cases}}
\newcommand{\beqs}{\begin{equation*}}
\newcommand{\eeqs}{\end{equation*}}

\newcommand{\mvec}{\mathrm{vec}}

\newcommand{\FF}{\mathbb{F}}

\newcommand{\PP}{\mathbb{P}}

\newcommand{\EE}{\mathbb{E}}
\newcommand{\RR}{\mathbb{R}}
\newcommand{\NN}{\mathbb{N}}

\newcommand{\HH}{\mathbb{H}}

\newcommand{\ind}{\mathbf{1}}
\newtheorem{ass}{\it{Assumption}}
\newtheorem{mydef}{\it{Definition}}
\newtheorem{lem}{\it{Lemma}}

\newtheorem{thm}{\it{Theorem}}

\newtheorem{prop}{\it{Proposition}}

\newtheorem{rem}{\it{Remark}}

\newtheorem{cl}{\it{Claim}}
\usepackage{hyperref}
\def\BibTeX{{\rm B\kern-.05em{\sc i\kern-.025em b}\kern-.08em
    T\kern-.1667em\lower.7ex\hbox{E}\kern-.125emX}}
\begin{document}

\title{Secure state estimation over Markov wireless communication channels}

\author{Anastasia Impicciatore$^1$, Anastasios Tsiamis$^2$, Yuriy Zacchia Lun$^1$, Alessandro D'Innocenzo$^1$, George J. Pappas$^3$
\thanks{$^1$ Department of Information Engineering, Computer Science and Mathematics, University of L'Aquila, L'Aquila 67100,  Italy (e-mails:{\tt\small anastasia.impicciatore@graduate.univaq.it}, {\tt\small yuriy.zacchialun@univaq.it}, {\tt\small alessandro.dinnocenzo@univaq.it}) }
\thanks{$^2$ Automatic Control Laboratory, ETH Zürich
 (e-mail: 
{\tt\small atsiamis@control.ee.ethz.ch})}
\thanks{$^3$ Department of Electrical and Systems Engineering
University of Pennsylvania, 200 South 33rd Street, Philadelphia, PA 19104, United States (e-mail: {\tt\small pappasg@seas.upenn.edu})}
}

\maketitle

\begin{abstract}
\noindent This note studies state estimation in wireless networked control systems with secrecy against eavesdropping. Specifically, a sensor transmits a system state information to the estimator over a legitimate user link, and an eavesdropper overhears these data over its link independent of the user link. Each connection may be affected by packet losses and is modeled by a finite-state Markov channel (FSMC), an abstraction widely used to design wireless communication systems. This paper presents a novel concept of optimal mean square expected secrecy over FSMCs and  delineates the design of a secrecy parameter requiring the user mean square estimation error (MSE) to be bounded and  eavesdropper MSE unbounded. We illustrate the developed results on an example of an inverted pendulum on a cart whose parameters are estimated remotely over a wireless link exposed to an eavesdropper.
\end{abstract}

\section{Introduction}
\noindent Wireless networked control systems (WNCSs) comprise spatially distributed networked sensors, actuators, and controllers providing closed-loop control over wireless communication media. These systems find applications in industrial automation, intelligent transportation, and smart grids, receiving considerable attention from industry and academia \cite{Park2018wireless,lu2015real}. The significant challenges of wireless connectivity, especially for the control applications, lie in the time-varying, unreliable and shared nature of this communication medium. The movement of people and objects in the propagation environment induces the shadow and small-scale fading that, paired with interference from other transmitters, causes information loss leading to performance and stability degradation \cite{Quevedo2013state}.
Furthermore, due to the shared nature of the wireless medium, other agents in the vicinity can overhear the content of transmissions, and there is often a need to protect systems from eavesdroppers \cite{Leong2017state,Chong2019atutorial}.

The current defense mechanisms against eavesdropping in WNCSs involve encryption-based tools, wireless physical-layer security methods, and control-theoretic approaches \cite{Schulze2021encrypted,Alexandru2021encrypted,%
Leong2019transmission,Tsiamis2020state}. Like \cite{Tsiamis2017state}, this paper relies on control theory to take advantage of the system dynamics to provide security guarantees by randomly withholding sensor information.
However, unlike any other control-theoretic contribution,
it does not consider the transmission over wireless links being modeled by an independent and identically distributed (i.i.d.) Bernoulli random variable or a time-homogeneous two-state Markov chain (MC) but by a finite-state Markov channel (FSMC). FSMC is an important model because
traditionally, wireless communication systems designers use this mathematical abstraction for modeling error bursts in fading channels to analyze and improve performance measures in the physical or media access control layers. Moreover, several receivers' channel state estimation and decoding algorithms rely upon FSMC models \cite{Sadeghi2008finite}.

In this work each agent (user or eavesdropper) estimates the process evolution of the Signal-to-Interference-plus-Noise Ratio (SINR) on its link, independently from the other. A finite-state MC (with more than two modes) approximates the SINR process over each link. A binary random variable standing for the outcome of the transmission is associated to each Markov mode, which determines the distribution of the binary random variable. The resulting FSMC allows for a tighter integration in the coupled design of the communication and estimation components of the WNCS.

Some procedures for control and estimation over packet dropping wireless links modeled by FSMCs can be related to the Markov jump linear systems (MJLSs) theory \cite{LUN2019stabilizability,Impicciatore2021optimal,de2020iet,%
de2020ifac,de2022cdc} generalizing the fundamental results based on i.i.d. Bernoullian assumptions \cite{Schenato2007foundations}. Nevertheless, most of the contributions on estimation and control over fading channels consider the two-state MC modeling a bursty packet erasure channel \cite{Mo2013LQG}. In this article, we choose a minimum mean square Markov jump filter instead of Kalman filter (see \cite{Mo2013LQG}), because the filter dynamics depends just on the current mode of the sensing channel (rather than on the entire past history of modes).
\subsection{Paper contribution}
\noindent This work brings the perfect expected secrecy notion in \cite{Tsiamis2017state} to FSMCs. In contrast to \cite{Tsiamis2017state}, we study secrecy over estimation filters, whose gains can be pre-computed offline. The original notion of perfect expected secrecy requires implementations of the Kalman filter. However, under FSMCs, any offline computation of the Kalman filter gains would require a combinatorially increasing, with the time-horizon, amount of memory. For this reason, here, we consider an alternative practical notion of expected secrecy; we consider the minimum MSE, but over filters with a finite number of offline-computed gains. 

In particular, in this paper, we require the eavesdropper MSE to grow unbounded, while the user MSE remains bounded. This new definition, requires us to adopt a different approach and utilize tools for the stability analysis of MJLSs \cite{COSTA2005,LUN2019stabilizability}.

We employ a secrecy mechanism, which, similar to \cite{Tsiamis2017state}, randomly withholds information with some probability. We prove that by properly tuning the withholding probability, we can achieve expected secrecy if an only if there is channel disparity between the user and the eavesdropper, i.e., the user has a higher probability of packet reception on average. Finally, we also provide novel covariance lower bounds for the eavesdropper MSE. Such a lower bound could be used as a guide to tune the withholding probability of the secrecy mechanism.

\subsection{Notation and preliminaries}
In the following, \mbox{$\mathbb{N}_{0}$} denotes the set of non-negative 
integers, \mbox{$\mathbb{R}$} denotes the set of reals, while 
\mbox{$\mathbb{F}$} indicates the set of either real or complex numbers.
The absolute value of a number is denoted by \mbox{$|\cdot|$}.
For positive integers \mbox{$r$} and \mbox{$s$}, the symbol
\mbox{$\mathbf{O}_{r}$} denotes the vector containing all zeros of length 
\mbox{$r$}. \mbox{$\mathbb{I}_{r}$} indicates the identity matrix 
of size \mbox{$r$}, while \mbox{$\mathbb{O}_{r}$} represents the matrix 
of zeros of size \mbox{$r\times r$}.
Consider a vector \mbox{$x\in\RR^{r}$} and a matrix \mbox{$K\in\RR^{r\times s}$}. The transposition is denoted by \mbox{$x'$} and \mbox{$K'$}, the complex conjugation is denoted by \mbox{$\overline{x}$} and \mbox{$\overline{K}$}, the conjugate transposition is denoted by \mbox{$x^*$} and \mbox{$K^*$}, respectively. 
\mbox{$\FF_*^{r\times r}$} and \mbox{$\FF_+^{r\times r}$} represent 
the sets of Hermitian and positive semi-definite matrices, respectively. For any 
positive integers \mbox{$N,r$}, a sequence of matrices \mbox{$K_m$}, \mbox{ $m=1,\ldots,N$}, denoted by \mbox{$\mathbf{K}=\left[K_m\right]_{m=1}^{N}$}, we define the following sets of sequences of Hermitian matrices and of positive semi-definite Hermitian matrices:
\begin{align*}
\HH^{Nr,*}&\triangleq\{\mathbf{K}=
[K_m]_{m=1}^{N};K_m\in\FF^{r\times r}_*, m=1,\ldots,N\},\\
\HH^{Nr,+}&\triangleq\{\mathbf{K}\in\HH^{Nr,*};
K_m\in\FF^{r\times r}_+, m=1,\ldots,N\}.
\end{align*}
We denote by \mbox{\small $\rho(\cdot)$} the spectral radius of a square matrix, i.e., the largest absolute value of its eigenvalues, and by
\mbox{\small $\left\| \cdot \right\|$} either any vector norm or any matrix norm. The operator \mbox{$\mvec\left(\cdot\right)$} denotes the vectorization of a matrix. Given \mbox{$\mathbf{K}=\left[ K_m\right]_{m=1}^{N}$}, 
\begin{align*}
\mvec^2\left(\mathbf{K}\right)=\left[\mvec\left(K_1\right),\ldots, \mvec\left(K_{N}\right)\right]'.
\end{align*}
\noindent Let \mbox{$\otimes$}, \mbox{$\bigoplus$} denote the Kronecker product defined in the usual way (see for example \cite{brewer1978kronecker}) and the direct sum, respectively.
\subsection{Paper organization}
\noindent The paper is organized as follows. The problem formulation is presented in Section \ref{sec:Pf}. The optimal mean square expected secrecy is provided by Section \ref{sec:opt_ms_sec} and the main result is shown in Section \ref{sec:main_result}. An eavesdropper characterization is presented in Section \ref{sec:eaves_char}. Finally, an example can be found in Section \ref{sec:Example}. Proofs and technical results are reported in the appendix. 
\section{Problem formulation}\label{sec:Pf}
\noindent The following discrete-time linear system describes the plant
\begin{align}\label{eq:system}
\bc  x(k+1)=Ax(k)+w(k),\\
  y(k)= Lx(k)+v(k),	\ec
\end{align}
\noindent where \mbox{$x(k)\in\RR^{n_x}$} is the state and \mbox{$y(k)\in\RR^{n_y}$} is the system's output, while \mbox{$k\in\NN$} is the (discrete) time. The signals \mbox{$w(k)\in\RR^{n_x}$} and \mbox{$v(k)\in\RR^{n_y}$} are the process and measurement noise respectively: \mbox{$w(k)$} and \mbox{$v(k)$} are i.i.d. independent Gaussian random variables with zero mean and covariance matrices \mbox{$Q,R\succ 0$} respectively. The initial state \mbox{$x_0$} is Gaussian with zero mean and covariance matrix \mbox{$\Sigma_0\succ 0$}.
\begin{figure}
    \centering
          \includegraphics[width=\linewidth]{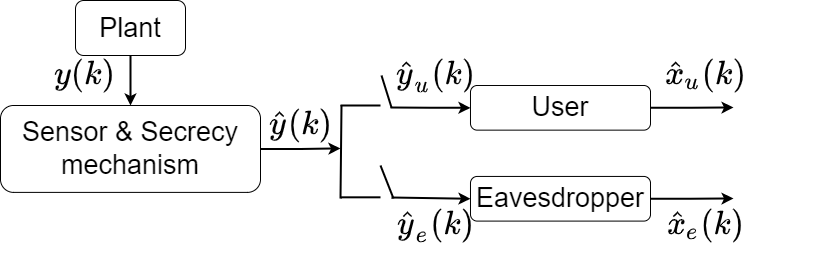}
        \label{fig:Architecture_Scheme}
        \caption{Remote estimation architecture.}
    \end{figure}

\begin{ass}\label{ass:mat_A}
The system described by \eqref{eq:system} is unstable, i.e., \mbox{$\rho\left( A\right)>1$}.
\end{ass}
\noindent Even without eavesdroppers,
estimation of unstable open-loop systems has been a problem of independent interest in control
systems (see \cite{Sinopoli2004Kalman} for instance). The ultimate goal is to close the loop and apply control, but first, estimation of the open-loop system should be studied. Besides, if the system is stable, the eavesdropper can predict the state with more accuracy, even without eavesdropping, since the state remains close to the origin.
\subsection{Secrecy mechanism} 
\noindent We adopt the secrecy mechanism introduced in \cite{Tsiamis2017state}: the sensor transmits the output \mbox{$y(k)$} with probability \mbox{$\lambda\in\left[0,1\right]$} and it transmits no information (denoted by the symbol \mbox{$\epsilon$}) with probability \mbox{$1-\lambda$}. Formally,  
\begin{align}\label{eq:privacy_mec}
\hat{y}(k)=\bc y(k) & \text{if }	\nu(k)=1\\
\epsilon & \text{if } \nu(k)=0
\ec \quad \forall k\geq 0,
\end{align}
\noindent \mbox{$\nu(k)$} is the outcome of the secrecy mechanism, represented by a binary random variable with secrecy parameter \mbox{$\lambda$} defined as follows: \mbox{$\mathbb{P}\left(\nu(k)=1\right)=\lambda$}, \mbox{$\mathbb{P}\left(\nu(k)=0\right)=1-\lambda$}. In the rest of this note, we will use the subscript $i$ to indicate an agent operating at the receiver's end.\\\noindent Formally, \mbox{$i\in\{u,e\}$}, where \mbox{$u$} refers to the user and \mbox{$e$} marks the eavesdropper. We will also call the \mbox{$i$}-th link the wireless link between the plant and the agent \mbox{$i$}.   
\subsection{Wireless link}
\noindent Let \mbox{$\hat{y}_i(k)$} denote the measurement received by the agent \mbox{$i$} at time \mbox{$k\in\NN$}. The general model for the agent's link is
\begin{align}\label{eq:channel_model_est}
\hat{y}_i(k)=\bc
\hat{y}(k) & \text{if }\xi_i(k)=1\\
\epsilon  & \text{if }\xi_i(k)=0\ec  \quad \forall k\geq 0,
\end{align}
\noindent where \mbox{$\epsilon$} means no information. The  \mbox{$i$}-th link is modeled exploiting the mathematical abstraction provided by the finite-state Markov channel (FSMC). The packet arrival process on the \mbox{$i$}-th link is described by
the process \mbox{$\xi_i(k)$}, \mbox{$k\in\NN$}: its value is \mbox{$\xi_i(k)=0$} if the packet is lost, \mbox{$\xi_i(k)=1$} if the packet is correctly delivered. The process \mbox{$\xi_i(k)$} is a binary random variable and the probability of having a packet loss or a correct packet transmission over the link \mbox{$i$} depends on the  SINR. The SINR is determined by physical phenomena and model parameters (see \cite{Lun2020ontheimpact}) such as path loss, shadow fading, interference and also the nature of the environment (domestic or industrial). \mbox{$\text{SINR}$} is a stochastic process 
and can be approximated by a finite state MC, denoted as \mbox{$\eta_i(k)$}. Let \mbox{$\mathbb{S}\triangleq\{1,\ldots,N\}$} be the index set of the Markov modes, then \mbox{$\eta_i(k)\in\{s_{i,m}\}_{m\in\mathbb{S}}$} (see \cite{Lun2020ontheimpact}). Each agent \mbox{$i$} estimates the SINR process during a learning phase and thus, it knows the number of modes, the transition probabilities, and the probability distribution of the MC \mbox{$\eta_i$}. For each mode of \mbox{$\eta_i(k)$}, the value of \mbox{$\xi_{i}(k)$} can be either zero or one with certain probabilities. For \mbox{$m\in\mathbb{S}$}, let the variable \mbox{$\hat{\gamma}_{i,m}$} denote the probability that \mbox{$\xi_i(k)=1$}, given the mode of the MC \mbox{$\eta_i(k)$} for \mbox{$k\in\NN$}, 
\begin{align*}
\PP\left(\xi_i(k)=1\mid\eta_i(k)=s_{i,m}\right)&=\hat{\gamma}_{i,m},\\
\PP\left(\xi_i(k)=0\mid\eta_i(k)=s_{i,m}\right)&=1-\hat{\gamma}_{i,m}.
\end{align*}
\noindent When a packet loss on the \mbox{$i$}-th link has occurred or the secrecy mechanism has withheld the output information, the agent \mbox{$i$} interprets the system state message as lost. Specifically, based on error detection and correction mechanisms the receiver decides whether the packet is \mbox{$\epsilon$} and should be dropped. For most communication protocols receiver also performs SNIR estimation for each received packet and thus the agent \mbox{$i$} always knows the mode of the \mbox{$i$-th} link \mbox{$\eta_i(k)$}. 
\\\noindent 
\noindent From \eqref{eq:privacy_mec}-\eqref{eq:channel_model_est} the received measurement \mbox{$\hat{y}_i(k)$} is different from \mbox{$\epsilon$} if and only if \mbox{$\nu(k)\xi_i(k)=1$}. We define the variable \mbox{$\varphi_i(k)\triangleq \nu(k)\xi_i(k)$}. For \mbox{$i\in\{u,e\}$}, \mbox{$m\in\mathbb{S}$},
\begin{align*}
\PP\left(\varphi_{i}(k)=1|\eta_i(k)=s_{i,m}\right)&=\lambda\hat{\gamma}_{i,m},\\
\PP\left(\varphi_{i}(k)=0|\eta_i(k)=s_{i,m}\right)&=1-\lambda\hat{\gamma}_{i,m} .
\end{align*}
The information set available to the agent \mbox{$i$} at time \mbox{$k\in\NN$} is given by  \mbox{$\mathcal{F}_i(k)=\{\left(\hat{y}_i(t)\right)_{t=0}^k, \left(\varphi_i(t)\right)_{t=0}^k, \left(\eta_i(t)\right)_{t=0}^{k}\}$}.
\begin{rem}\label{rem:info_set}
Observing the information set \mbox{$\mathcal{F}_i(k)$} and recalling  the definition of \mbox{$\hat{y}_i(k)$} in \eqref{eq:channel_model_est} and the secrecy mechanism \eqref{eq:privacy_mec}, it is straightforward to see that  
with the knowledge of \mbox{$\hat{y}_i(k)$} and \mbox{$\varphi_i(k)$}, the agent \mbox{$i$} is aware of \mbox{$y(k)$}. 
\end{rem}

\subsection{Probabilistic framework}
\noindent Let \mbox{$\pi_{i,m}(k) \triangleq \PP\left(\eta_i(k)=s_{i,m}\right)$}, with \mbox{$0<\pi_{i,m}(k)<1$}, for any \mbox{$k$}, for \mbox{$m\in\mathbb{S}$}, \mbox{$i\in\{u,e\}$}. A Transition Probability Matrix (TPM) associated with the MC \mbox{$\eta_i(k)$} is denoted by \mbox{$P_i\triangleq\left[ p_{i,mn}\right]_{m,n=1}^{N}$},
\begin{align*}
p_{i,mn}\!=\!\PP\left(\eta_i(k+1)\!=\!s_{i,n}|\eta_i(k)\!=\!s_{i,m}\right), \sum\limits_{n=1}^{N} p_{i,mn}\!=\!1.
\end{align*}
\noindent Similarly to \cite[Sec.~5.3]{COSTA2005}, we make the following technical assumptions (with \mbox{$i\in\{u,e\}$} and \mbox{$k\in\NN$}):
\begin{itemize}
\item[$i)$] the initial conditions \mbox{$x_0, \eta_{i,0}$} are independent random variables,
\item[$ii)$] the white  noise sequences \mbox{$\{w(k)\}$} and \mbox{$\{v(k)\}$} are independent of the initial 
conditions \mbox{$(x_0,\nu(0))$} and of the processes \mbox{$\xi_{i}(k)$}, for any  
discrete-time \mbox{$k\in\NN$},
\item[$iii)$] the MCs 
\mbox{$\{\eta_i(k)\}$} and the noise sequences \mbox{$\{w(k)\}$} and \mbox{$\{v(k)\}$} are independent,
\item[$iv)$] the MCs 
\mbox{\small $\{\eta_i(k)\}$} are ergodic, with steady state probability
distributions 
\mbox{ $\pi_{i,m}^{\infty} = \lim\limits_{k\to\infty}\pi_{i,m}(k)$}, \mbox{ $m\in\mathbb{S}$}. 
\end{itemize}
\noindent This work aims to design an estimator of the class of mean square Markov jump filters (see \cite[Ch.5.3]{COSTA2005}) together with a secrecy mechanism, such that the user MSE remains bounded, while the eavesdropper MSE is unbounded. The formal guarantees of this secrecy notion can be found in Section \ref{sec:opt_ms_sec}, see Definition \ref{def:perfect_expected_secrecy}. Here we introduce the variables \mbox{$\psi_i$} and \mbox{$\zeta_i$} that will be useful for the statement and the proof of our main result. For \mbox{$i\in\{u,e\}$}, \mbox{$\psi_i$} denotes the average probability of intercepting a measurement on the \mbox{$i$}-th link when \mbox{$\lambda=1$}, \mbox{$\zeta_i$} is the average probability of intercepting a measurement for \mbox{$\lambda\in\left[0,1\right)$}. Formally, \mbox{$\psi_i\triangleq\sum_{m=1}^{N}\pi_{i,m}^{\infty}\hat{\gamma}_{i,m}$}, \mbox{$\zeta_i\triangleq\psi_i\lambda$}.
\section{Optimal mean square expected secrecy}\label{sec:opt_ms_sec}
\noindent We present the infinite horizon minimum mean square Markov jump filter (see \cite[Ch.5.3]{COSTA2005} with the estimation technique provided by an estimator called current estimator \cite[Ch. 8.2.4]{Franklin1997}). Specifically, the estimator provides at each step a model prediction obtained from the estimated state at the previous step. This prediction is corrected by the current measurement received \mbox{$\hat{y}_i(k)$}.
\begin{rem}\label{rem:difference_with_Kalman}
It is well known that for the case in which the information on the output of the system and on the MC are available at each time step \mbox{$k\in\NN$}, the best linear estimator of \mbox{$x(k)$} is the Kalman filter (see \cite[Remark 5.2]{COSTA2005}). An offline computation of the Kalman filter is inadvisable here as pointed out in \cite{Ji1990JumpLQ}. The reason is that the solution of the difference Riccati equation and the time varying Kalman gain are sample path dependent and the number of sample paths grows exponentially in time. On the other hand, an online computation of the Kalman filter requires online matrix inversions which might require a lot of computation. For this reason, we consider a different class of estimators, for which we can pre-compute the filtering gains offline. This allows us to avoid online matrix inversions, thus, reducing the computational burden.  
\end{rem}
\noindent Recalling that the agent \mbox{$i$} receives a quantity that is different from \mbox{$\epsilon$} if and only if \mbox{$\varphi_i(k)=1$} (see Remark \ref{rem:info_set}), the current estimated state dynamics can be written as follows (see also \cite[eq. (8.33)-(8.34)]{Franklin1997}), for \mbox{$i\in\{u,e\}$}:
\begin{align}
\hat{x}_i(k)&=\overline{x}_i(k)-\varphi_i(k) \widehat{M}_{i,\eta_i(k)} \left[y(k)- L\overline{x}_i(k)\right],\label{eq:current_est1}\\
\overline{x}_i(k+1)&=A \hat{x}_i(k),\label{eq:current_est2}
\end{align}
\noindent where \mbox{$\widehat{M}_{i,\eta_i(k)}$} is the mode-dependent filtering gain, whose explicit expression can be found later in \eqref{eq:diff_Riccati_2}. Since the filtering gain depends on the mode of the MC at time \mbox{$k$}, and the MC has a given finite set of modes, it can be computed offline (see Remark \ref{rem:filt_gain}).
\noindent 
From \eqref{eq:current_est1}-\eqref{eq:current_est2}, by defining the error as \mbox{$\tilde{e}_i(k)=x(k)-\overline{x}_i(k)$}, \mbox{$i\in\{u,e\}$},
\begin{align}\label{eq:current_error1}
\tilde{e}_i(k+1)=&\left(A  + \varphi_i(k) A \widehat{M}_{i,\eta_i(k)} L \right)\tilde{e}_i(k) + w(k) 
\nonumber\\&+ \varphi_i(k) A\widehat{M}_{i,\eta_i(k)} v(k),
\end{align}
\noindent see also \cite[eq. (8.36)]{Franklin1997}.
\begin{rem}\label{rem:markov_jump_error}
The error system described by \eqref{eq:current_error1} is a discrete-time MJLS (see for instance \cite{COSTA2005}).
\end{rem}
\noindent The notation presented in \cite{COSTA2005} is adopted here: for \mbox{$i\in\{u,e\}$}, \mbox{$m\in\mathbb{S}$}, let us define \mbox{$\mathbf{Z}_i(k)\triangleq\left[Z_{i,m}(k)\right]_{m=1}^{N}\in\HH^{Nn_x,+}$}, 
\begin{align*} 
Z_{i,m}(k)\triangleq \EE\left[\tilde{e}_i(k) \tilde{e}_i^{*}(k)\ind_{\{\eta_i(k)=s_{i,m}\}}\right],
\end{align*}
\noindent with \mbox{$\ind_{\{\eta_i(k)=s_{i,m}\}}$} denoting the indicator function defined in the usual way.\\\noindent The MSE can be written as follows (see for instance \cite{COSTA2005}, \cite{Impicciatore2021optimal}), \mbox{$\EE\left[\tilde{e}_i(k)\tilde{e}_i^{*}(k)\right]=\sum_{m=1}^{N}Z_{i,m}(k)$}.\\\noindent
Given the MSE expression, we are ready to introduce the definition of optimal mean square expected secrecy over FSMCs.
\begin{mydef}[Secrecy over FSMCs]\label{def:perfect_expected_secrecy}
\noindent Given the system described by \eqref{eq:system} and the FSMCs \eqref{eq:channel_model_est}, we say that a secrecy mechanism \eqref{eq:privacy_mec} achieves optimal mean square expected secrecy over FSMCs if and only if, for any initial condition \mbox{$\mathbf{Z}_i(0)\in\HH^{Nn_x,+}$}, \mbox{$i\in\{u,e\}$}, both of the following conditions hold: \mbox{$\lim \nolimits_{k\to\infty} \bold{tr}\left\{\EE\left[\tilde{e}_u(k) \tilde{e}_u^*(k)\right]\right\}<\infty$}, \mbox{$\lim\nolimits_{k\to\infty}\bold{tr}\left\{\EE\left[\tilde{e}_e(k) \tilde{e}_e^*(k)\right]\right\}=\infty $}.
\end{mydef}
\begin{ass}\label{ass:basis_mec}
If the secrecy mechanism \mbox{$\hat{y}(k)=y(k)$} is employed for all \mbox{$k\geq 0$}, i.e., when \mbox{$\lambda=\frac{\zeta_u}{\psi_u}=1$}, the user MSE is bounded, i.e., \mbox{$\lim \nolimits_{k\to\infty} \bold{tr}\left\{\EE\left[\tilde{e}_u(k) \tilde{e}_u^*(k)\right]\right\}<\infty$}, for any initial condition \mbox{$\mathbf{Z}_u(0) \in\HH^{Nn_x,+}$}.
\end{ass}
\noindent The following operator is instrumental for the presentation of the Algebraic Riccati equation and for the technical results exploited in the proof of the main theorem. Let us define the operator \mbox{$\mathcal{X}_{\lambda}:\FF_+^{n_x\times n_x}\times\RR^+\times\RR^+\to\FF_+^{n_x\times n_x}$}, for \mbox{$\lambda\in\left[0,1\right]$}, \mbox{$X\in\FF_+^{n_x\times n_x}$}, \mbox{$\alpha>0$}, \mbox{$\phi\in\RR^+$},
\begin{align}\label{eq:mathcalX_def}
&\mathcal{X}_{\lambda}\left(X,\alpha,\phi\right)\triangleq \left(1-\lambda\phi\right)\{A X A^*  + \alpha Q\}+\nonumber\\
& \lambda \phi \left(A X A^*  + \alpha Q - A X L^* \left(L X L^* + \alpha R\right)^{-1} L X A^*\right).
\end{align}
\begin{prop}\label{prop:recursive_covariance}
Consider the error system described by \eqref{eq:current_error1}. Under Assumption \ref{ass:basis_mec}, for \mbox{$m,n\in\mathbb{S}$}, \mbox{$i\in\{u,e\}$}, the filtering coupled algebraic Riccati equations (CAREs) are 
\begin{align}
Z_{i,n} &=\sum\limits_{m=1}^{N} p_{i,mn}\mathcal{X}_{\lambda}\left(Z_{i,m},\pi_{i,m}^{\infty},\hat{\gamma}_{i,m} \right),\label{eq:diff_Riccati_1}\\
\widehat{M}_{i,m} &= -Z_{i,m} L^* \left(L Z_{i,m} L^* + \pi_{i,m}^{\infty} R\right)^{-1}.\label{eq:diff_Riccati_2}
\end{align}
\noindent
\end{prop}
\begin{proof}
See appendix.
\end{proof}
\begin{rem}\label{rem:filt_gain}
The filtering gain can be computed offline from the minimization of the MSE, according to the procedure shown in \cite{Impicciatore2021optimal}. Particularly, each agent \mbox{$i$} knows the matrices of the system, as well as the mode of the MCs \mbox{$\eta_i$}. Formally, for \mbox{$m\in\mathbb{S}$}, the filtering gain \mbox{$\widehat{M}_{i,m}$} is given by \eqref{eq:diff_Riccati_2}, where \mbox{$Z_{i,m}$} is the solution of \eqref{eq:diff_Riccati_1}.
\end{rem}
\section{Main result}\label{sec:main_result}
\noindent In this section we present necessary and sufficient conditions concerning the FSMCs probabilities such that optimal mean square expected secrecy is guaranteed.
\begin{thm}\label{thm:privacy_conditions}
Consider the system described by \eqref{eq:system}, the secrecy mechanism given by \eqref{eq:privacy_mec}, and FSMCs described by \eqref{eq:channel_model_est}. Under Assumption \ref{ass:mat_A} and Assumption \ref{ass:basis_mec}, the secrecy mechanism achieves optimal mean square expected secrecy over FSMCs if and only if \mbox{$\psi_u>\psi_e$}.\\\noindent
In particular, there exists a probability \mbox{$\zeta_c\in\left[0,1\right)$} such that optimal mean square expected secrecy is guaranteed if and only if the probability \mbox{$\lambda$} in the secrecy mechanism satisfies the following inequalities
\begin{align*} 
\frac{\zeta_c}{\psi_u}<\lambda\leq \min\left\{\frac{\zeta_c}{\psi_e},1\right\}.
\end{align*}
\end{thm}
\begin{rem}\label{rem:psi_u_psi_e}
The inequality \mbox{$\psi_u>\psi_e$} is a reasonable condition for secrecy in many cases of interest. Indeed, it is plausible that the propagation environment leads to an average probability of intercepting the measurement over the eavesdropper link, \mbox{$\psi_e$}, which is strictly less than \mbox{$\psi_u$}, for instance because the eavesdropper might be further away from the source.
\end{rem}
\begin{proof}
Let us show the sufficiency part.  Consider for \mbox{$n\in\mathbb{S}$}, \mbox{$i\in\{u,e\}$}, \mbox{$k\in\NN$}, the following equality, 
\begin{align*}
Z_{i,n}(k+1) &=\sum\limits_{m=1}^{N} p_{i,mn}\mathcal{X}_{\lambda}\left(Z_{i,m}(k),\pi_{i,m}(k),\hat{\gamma}_{i,m} \right).
\end{align*}
\noindent Under Assumption \ref{ass:basis_mec}, if \mbox{$\lim \limits_{k\to\infty} \bold{tr}\left\{\EE\left[\tilde{e}_e(k) \tilde{e}_e^*(k)\right]\right\}= +\infty$}, we can choose \mbox{$\lambda=1$}. Otherwise, since Assumptions \ref{ass:mat_A}-\ref{ass:basis_mec} hold, by Lemma~\ref{lem:convergence_forall}, for any \mbox{$\mathbf{Z}_0 \in\HH^{Nn_x,+}$}, \mbox{$m,n\in\mathbb{S}$}, \mbox{$i\in\{u,e\}$},
\begin{align}
&\lim\limits_{k\to\infty} \bold{tr}\left\{Z_{i,n}(k)\right\}=+\infty, \quad \text{for } 0\leq \lambda\leq \frac{\zeta_c}{\psi_i},\label{eq:lem_conv_forall1_applied}\\
& \lim\limits_{k\to\infty} \bold{tr}\left\{Z_{i,n}(k)\right\} < \infty, \quad \text{for }  \frac{\zeta_c}{\psi_i}<\lambda\leq 1.  \label{eq:lem_conv_forall2_applied}
\end{align}
\noindent This implies that the probability \mbox{$\lambda$} in the secrecy mechanism should be designed such that \mbox{$\lambda>\zeta_c/\psi_u$}, in order to guarantee \eqref{eq:lem_conv_forall2_applied} for the user MSE. Since the user MSE is bounded by assumption when \mbox{$\lambda=1$}, \mbox{$\psi_u \times 1 >\zeta_c $}, and thus \mbox{$\psi_u>\zeta_c$} implying \mbox{$\zeta_c/ \psi_u <1$}. 
\\\noindent
Consider now the eavesdropper MSE. The secrecy parameter \mbox{$\lambda$} should be chosen sufficiently small such that the inequality \mbox{$\lambda\leq \zeta_c/\psi_e$} is satisfied. Therefore, by choosing \mbox{$\lambda$} satisfying the following inequalities, \mbox{$\zeta_c/\psi_u<\lambda\leq\min\left\{\zeta_c/\psi_e,1\right\}$}, the secrecy mechanism guarantees optimal mean square expected secrecy over FSMCs by Lemma \ref{lem:convergence_forall}.\\\noindent 
Notice that the interval \mbox{$\left(\zeta_c/\psi_u,\min\left\{\zeta_c/\psi_e,1\right\}\right]$} is nonempty: \mbox{$\zeta_c/\psi_u<1$}, and \mbox{$\psi_u>\psi_e$} implies that \mbox{$\zeta_c/\psi_u<\zeta_c/\psi_e$}. \\
\noindent Let us show the necessity part. If the optimal mean square expected secrecy over FSMCs is achieved by the secrecy mechanism in \eqref{eq:privacy_mec}, by Lemma \ref{lem:convergence_forall} \mbox{$\zeta_c/\psi_u<\lambda\leq 1$} and \mbox{$\lambda\leq \zeta_c/\psi_e $}, implying \mbox{$\zeta_c/\psi_u<\lambda\leq \zeta_c/\psi_e$}.\\
\noindent Consequently, \mbox{$\lambda \psi_e< \lambda\psi_u$}, and finally \mbox{$\psi_e < \psi_u$}.\\
\noindent The proof of the theorem is complete.
\end{proof}
\section{Eavesdropper characterization}\label{sec:eaves_char}
Given the propagation environment, a designer can deduce possible positions of eavesdroppers, decide which are of the most concern, and derive an eavesdropper's TPM. 
\\\noindent In this section, we provide link quality constraints used to design the secrecy mechanism attempting to increase the eavesdropper MSE to infinity. More specifically, if the eavesdropper TPM \mbox{$P_e$} is known, the designer is able to construct the matrix \mbox{$\mathcal{A}_e$}, defined as follows,
\begin{align*}
\mathcal{A}_e\triangleq\left[P'_e\otimes \mathbb{I}_{n_x^2}\right]\left[\bigoplus_{m=1}^{N}\left(1-\lambda\hat{\gamma}_{e,m}\right)\left(\overline{A}\otimes A\right)\right].
\end{align*}
\noindent For \mbox{$\mathbf{V}=\bmat V_m\emat_{m=1}^{N} \in\HH^{Nn_x,*}$} define for \mbox{$n\in\mathbb{S}$},
\begin{align*}
\mathcal{S}_{e,n}\left(\mathbf{V}\right)\triangleq \sum\limits_{m=1}^{N}p_{e,mn}\left(1-\lambda\hat{\gamma}_{e,m}\right)A V_{m} A^* +\pi^{\infty}_{e,n}Q.
\end{align*}
\noindent  As we prove in the next proposition, it turns out that the operator \mbox{$\mathcal{S}_{e,n}$} defined above provides a lower bound to the eavesdropper MSE, under the estimator defined in \eqref{eq:current_est1}. Hence, we can use the above recursion to test whether the eavesdropper has MSE that increases to infinity.
%
%
\begin{prop}\label{prop:Lyap_eavesdropper}
Consider the system described by \eqref{eq:system} and the secrecy mechanism \eqref{eq:privacy_mec}. The following statements hold, for \mbox{$n\in\mathbb{S}$},
\begin{itemize}
\item  If \mbox{$\rho\left(\mathcal{A}_e\right)<1$}, then \mbox{$
\lim\limits_{k\to\infty} \bold{tr}\left\{Z_{e,n}(k)\right\}\geq \bold{tr}\left\{S_{e,n} \right\}$}, with \mbox{$S_{e,n}=\mathcal{S}_{e,n}\left(\mathbf{S}_e\right)$}, \mbox{$\mathbf{S}_e=\left[S_{e,n}\right]_{n=1}^{N}\in\HH^{Nn_x,+}$}.
\item If \mbox{$\rho\left(\mathcal{A}_e\right)\geq 1$}, then 
\mbox{$\lim\limits_{k\to\infty}\bold{tr}\left\{Z_{e,n}(k)\right\}=+\infty$}.
\end{itemize}
\end{prop}
\begin{proof}
See appendix.
\end{proof}
\section{Example}\label{sec:Example}
\noindent This section examines an inverted pendulum on a cart \cite{franklin2009feedback} whose parameters are estimated remotely over a wireless link exposed to an eavesdropper. 
The considered cart's and pendulum's masses are \mbox{$0.5\ \mathrm{kg}$} and \mbox{$0.2\ \mathrm{kg}$}, inertia about
the pendulum's mass center is \mbox{$0.006\ \mathrm{kg}\cdot \mathrm{m}^2$}, distance from
the pivot to the pendulum's mass center is \mbox{$0.3 \ \mathrm{m}$}, coefficient
of friction for the cart is \mbox{$0.1$}. The discrete-time system has been obtained from discretization with sampling \mbox{$T_s=0.01\ \mathrm{s}$} and linearization of the dynamical continuous time nonlinear model around the unstable equilibrium points \mbox{$\mathrm{x}^*=0 \ \mathrm{m}$}, \mbox{$\mathrm{\phi}^*=0 \ \mathrm{rad}$}. The resulting matrix \mbox{$A$} of the discrete-time system is such that \mbox{$\rho\left(A\right)\approx 1.1>1$}. This unstable plant evolves in open-loop.\\\noindent 
We recall that, when the propagation environment is known, a designer  can deduce which are the possible eavesdropping configurations allowing to overhear the user's messages. \\\noindent
Consider two independent wireless links: one link for the user, the other one for the eavesdropper. The formal mathematical description of a propagation environment accounts for different parameters. The two main parameters we refer in this description are the transmitter/receiver couple and the transmitter/interferer couple: the transmitter/receiver couple is the couple of interest, while the transmitter/interferer couple models some interference that affects the propagation environment and that characterize both the user and  the eavesdropper wireless link. Let \mbox{$d_u$} and \mbox{$d_e$} denote the distances of the couple of interest for the user and for the eavesdropper, respectively. Let \mbox{$\tilde{d}_u$} and \mbox{$\tilde{d}_e$} denote the distances of the couple transmitter/interferer for the user and for the eavesdropper, respectively.\\\noindent
Consider the following scenario (see \cite{Sadeghi2008finite,LUN2019stabilizability}):  \mbox{$d_u=17\ \mathrm{m}$}, \mbox{$\tilde{d}_u=15 \ \mathrm{m}$}, \mbox{$d_e=19\ \mathrm{m}$}, \mbox{$\tilde{d}_e=13 \ \mathrm{m}$}. With this configuration for user and eavesdropper the secrecy parameter \mbox{$\lambda$} guaranteeing the optimal mean square expected secrecy over FSMCs belongs to the interval \mbox{$\left(0.26,0.48\right]$}, and the limit probability \mbox{$\zeta_c\approx 0.105$}. 
\begin{figure}
    \centering
        \includegraphics[width=\linewidth ,height=4cm]{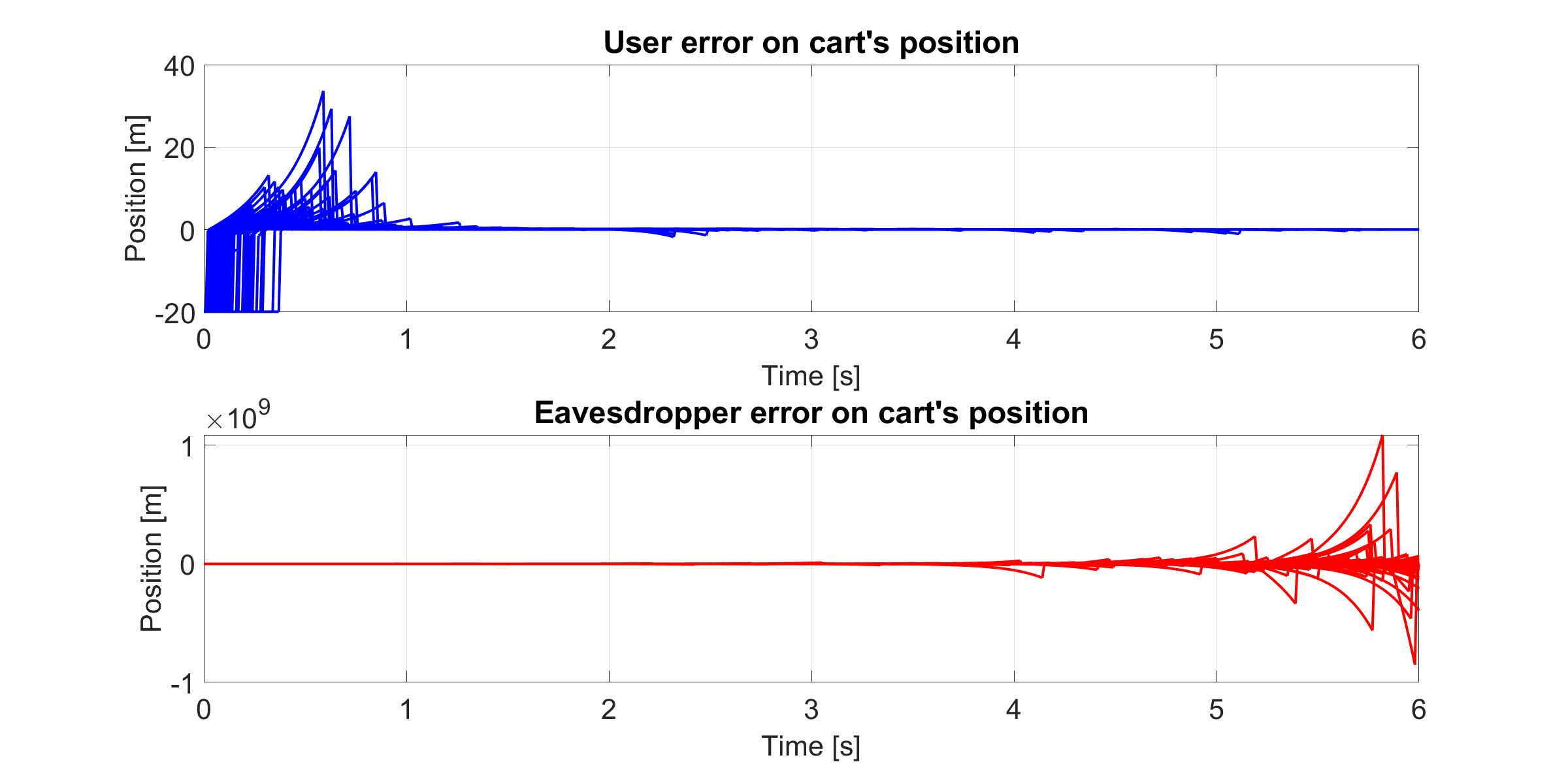}
        \caption{Error trajectories on cart's position for the user (blue trajectories) and for the eavesdropper (red trajectories) obtained with \mbox{$\lambda=0.3$}.}
        \label{fig:UserEaves_plots}
    \end{figure}
\\\noindent The results obtained in simulations are shown in Fig. \ref{fig:UserEaves_plots} and in Fig. \ref{fig:lambda03}. 
\\\noindent Fig. \ref{fig:UserEaves_plots} shows the error trajectories \mbox{$\tilde{e}_i(k)$}, \mbox{$i\in\{u,e\}$}, obtained from 1000 Montecarlo simulations for the user (blue lines) and the eavesdropper (red lines) with \mbox{$\lambda=0.3$}. As the reader can see, the user error trajectories have a convergent behavior, while the eavesdropper error trajectories diverge. Consider now Fig. \ref{fig:lambda03}, that reports the eavesdropper MSE (red line) and the user MSE (blue line) on cart's position with \mbox{$\lambda=0.3$}. The reader may notice that the eavesdropper MSE shows a worse behavior with respect to the user MSE: this is induced by the relation existing between the average probabilities of successfully receiving the system state message, \mbox{$\psi_e$} and \mbox{$\psi_u$}, over the eavesdropper and the user link, respectively. Particularly, in the reported example \mbox{$\psi_e=0.219$}, \mbox{$\psi_u=0.413$}, and thus  \mbox{$\psi_e<\psi_u$},  as required by Theorem \ref{thm:privacy_conditions}. More specifically, by comparing Fig.~\ref{fig:lambda1} (obtained without a secrecy mechanism) and Fig.~\ref{fig:lambda03} (obtained with the proposed secrecy mechanism), the reader may notice that the secrecy mechanism makes the eavesdropper MSE go to infinity, while the user MSE remains bounded (see Fig.~\ref{fig:lambda03}).
\section{Conclusion}\label{sec:Conclusion}
\noindent In this paper we considered secure state estimation over  Markov wireless communication channels. We bring the secrecy notion in \cite{Tsiamis2017state} to FSMCs, which requires re-definition of estimation problem and a novel technical procedure for deriving the secrecy conditions. Moreover, we design a secrecy mechanism satisfying the described formal requirements over 	FSMCs, and we show the effectiveness of our result in the example of an inverted pendulum on a cart whose parameters are estimated remotely over a wireless link exposed to an eavesdropper.
\begin{figure}
    \centering
        \includegraphics[width=\linewidth ,height=5cm]{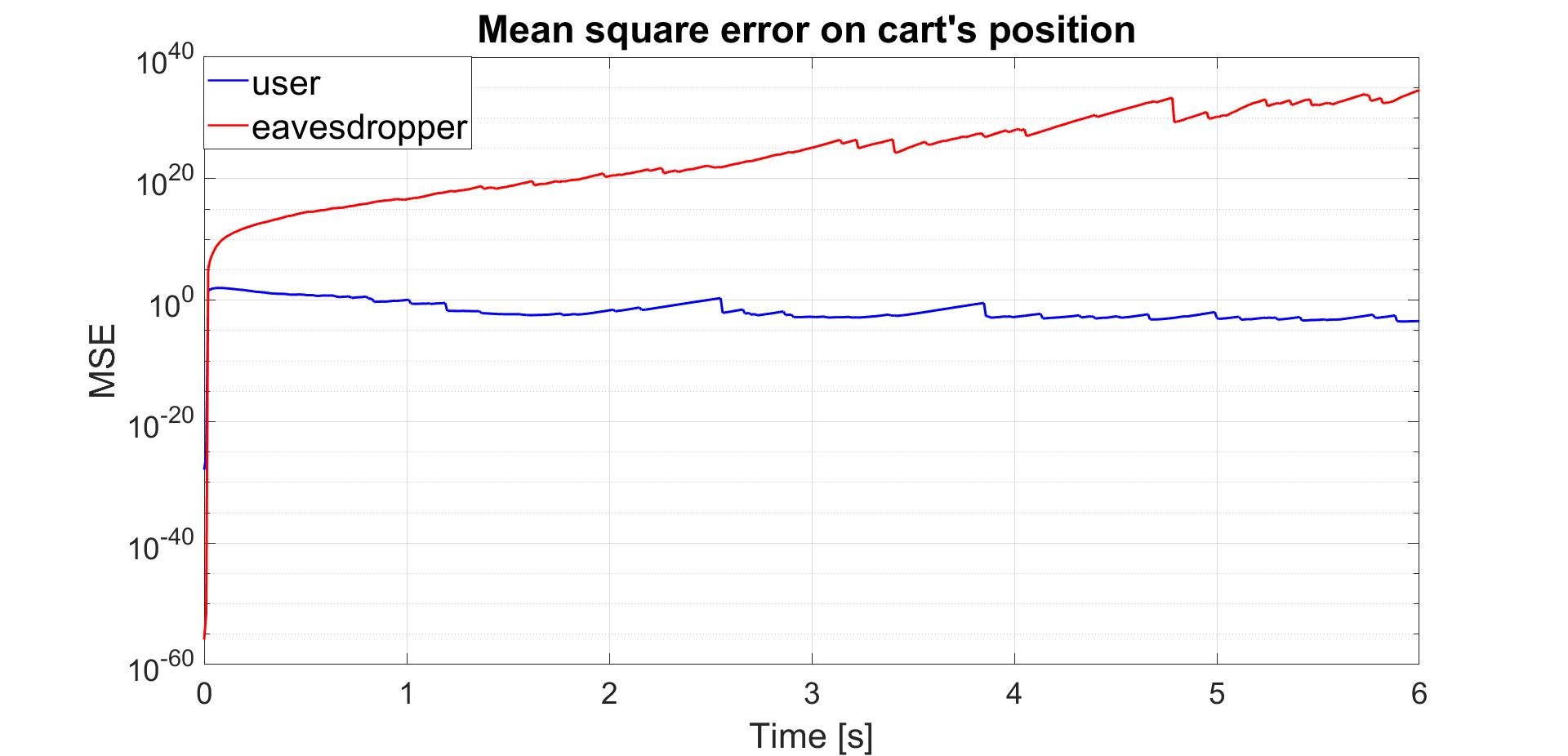}
        \caption{The figure reports the MSE on cart's position for the user (blue line) and for the eavesdropper (red line) with \mbox{$\lambda=0.3$}.} \label{fig:lambda03}
    \end{figure}
\begin{figure}
    \centering
        \includegraphics[width=\linewidth ,height=4cm]{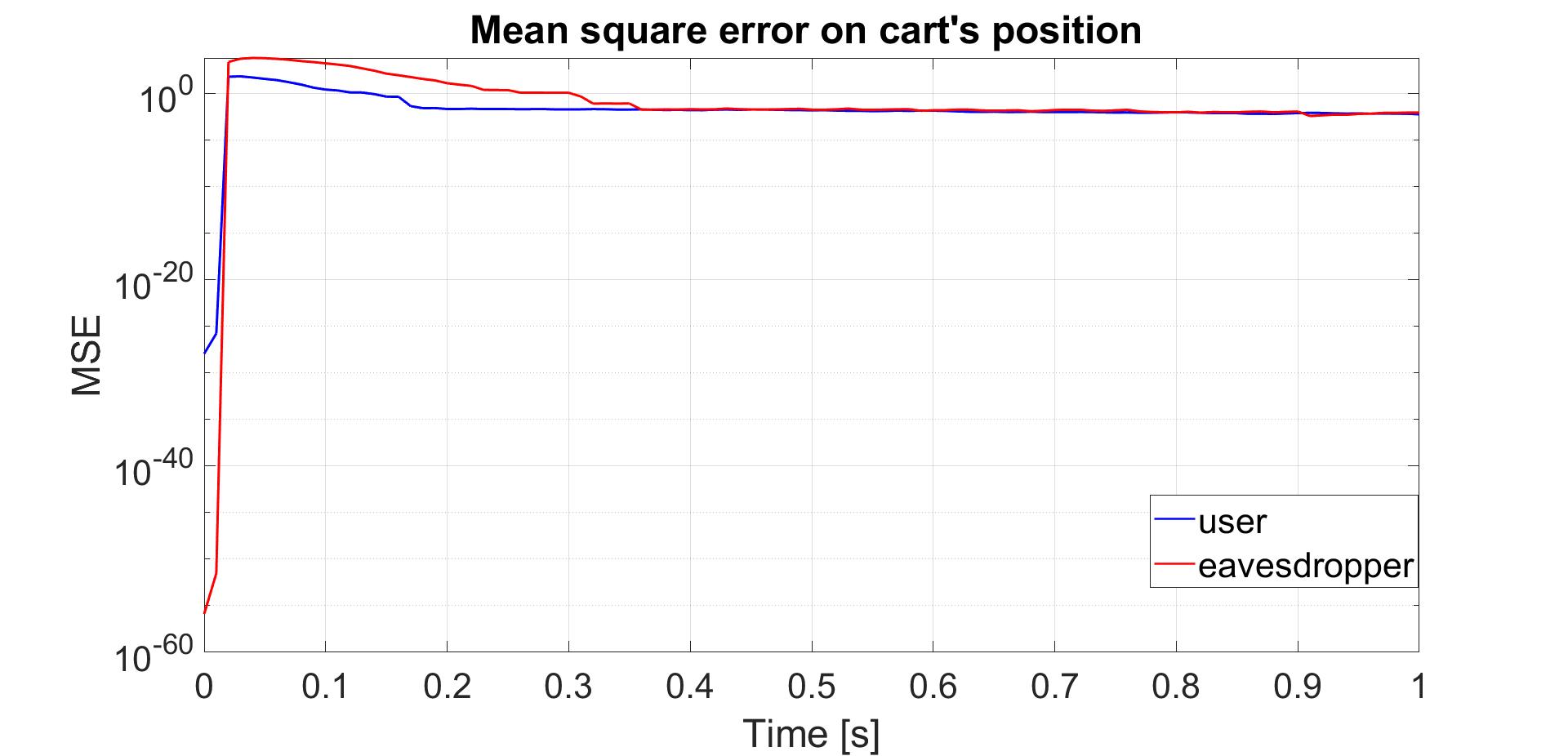}
        \caption{The figure reports the MSE on cart's position for the user (blue line) and for the eavesdropper (red line) with \mbox{$\lambda=1$}.} \label{fig:lambda1}
    \end{figure}

\appendix
\begin{proof}[Proof of Prop. \ref{prop:recursive_covariance}]
\noindent From the reasoning presented in \cite[Prop. A.23]{COSTA2005}, when there exists a mode-dependent filtering gain such that the MSE is bounded, from Assumption \ref{ass:basis_mec},
the difference Riccati equation is equivalent to the algebraic Riccati equation \eqref{eq:diff_Riccati_1}, and the mode-dependent filtering gain is given by \eqref{eq:diff_Riccati_2}.
\\\noindent The proof of the proposition is complete.   
\end{proof}
\noindent The following lemmas extend the result presented in \cite[Thm. 2]{Sinopoli2004Kalman} and they are instrumental for the proof of Thm. \ref{thm:privacy_conditions}.
\begin{lem}\label{lem:mathcalX_properties}
\noindent Consider for \mbox{$\lambda\in\left[0,1\right]$}, the operator \mbox{$\mathcal{X}_{\lambda}:\FF_+^{n_x\times n_x}\times\RR^+\times\RR^+\to\FF_+^{n_x\times n_x}$} defined 
for \mbox{$X\in\FF_+^{n_x\times n_x}$}, \mbox{$\alpha>0$}, \mbox{$\phi\in\RR^+$} in \eqref{eq:mathcalX_def}.\\\noindent
Then, the following conditions hold.
\begin{itemize}
\item[$(a)$]  Consider \mbox{$X,Y\in\FF_+^{n_x\times n_x}$}. If \mbox{$X\preceq Y$}, then \mbox{$\mathcal{X}_{\lambda}\left(X,\alpha,\phi\right)\preceq \mathcal{X}_{\lambda}\left(Y,\alpha,\phi\right)$}.
\item[$(b)$] Consider \mbox{$\hat{\lambda},\tilde{\lambda}\in\left[0,1\right]$}. If \mbox{$\hat{\lambda}\geq\tilde{\lambda}$}, then \mbox{$\mathcal{X}_{\hat{\lambda}}\left(X,\alpha,\phi\right)\preceq \mathcal{X}_{\tilde{\lambda}}\left(X,\alpha,\phi\right)$}.
\item[$(c)$] Consider the positive constant \mbox{$\beta>1$}. Then, \mbox{$\mathcal{X}_{\lambda}\left(\beta X,\alpha,\phi\right)\preceq \beta \mathcal{X}_{\lambda}\left(X,\alpha,\phi\right)$}.
\end{itemize} 
\end{lem}
\begin{proof}
\noindent Let us show that condition \mbox{$(a)$} is satisfied.\\
\noindent Define \mbox{$K_X\triangleq -XL^*\left(LXL^*+\alpha R\right)^{-1}$}.
 Then, the reader can easily verify that 
\begin{align*}
\mathcal{X}_{\lambda}&\left(X,\alpha,\phi\right)=
\left(1-\lambda\phi\right)\{A X A^*  + \alpha Q\}\nonumber\\
&+ \lambda \phi \Big\{\left(A+A K_X  L\right)X\left(A+A K_X L\right)^*+\alpha Q + \nonumber\\
&\alpha AK_X R K_X^* A^*\Big\}.
\end{align*}
\noindent By following an approach based on matrix derivatives, it can be shown that \mbox{$K_X$} minimizes the function \mbox{$\mathcal{X}_{\lambda}\left(X,\alpha,\phi\right)$} (see \cite[Lemma 1]{Sinopoli2004Kalman}), and thus recalling that \mbox{$X\preceq Y$}, we get \mbox{$\mathcal{X}_{\lambda}\left(X,\alpha,\phi\right)\preceq\mathcal{X}_{\lambda}\left(Y,\alpha,\phi\right)$}.
\\\noindent Let us show that condition \mbox{$(b)$} is satisfied. Assume that \mbox{$\hat{\lambda}\geq\tilde{\lambda}$}. Then,
\begin{align*}
\mathcal{X}_{\hat{\lambda}}&\left(X,\alpha,\phi\right) \preceq 
AXA^* + \alpha Q  \\
&- \tilde{\lambda}\phi AXL^*\left(LXL^*+ \alpha R\right)^{-1} L XA^* =\mathcal{X}_{\tilde{\lambda}}\left(X,\alpha,\phi\right).
\end{align*}
\noindent Let us show that condition \mbox{$(c)$} is satisfied. Assume that \mbox{$\beta>1$}. Then,
\begin{align*}
\mathcal{X}_{\lambda}&\left(\beta X,\alpha,\phi\right)\preceq \left(1-\lambda\phi\right)\left(\beta A X A^* + \beta\alpha Q\right)\\
&+\lambda\phi\Big(\beta AXA^* + \beta\alpha Q \\
&-\beta AXL^*\left(L\beta XL^*+ \beta\alpha  R\right)^{-1}L \beta XA^*\Big)\\
&=\beta \mathcal{X}_{\lambda}\left( X,\alpha,\phi\right),
\end{align*}
\noindent and thus, condition \mbox{$(c)$} is satisfied.\\
\noindent The proof of the lemma is complete. 
\end{proof}
\noindent The results presented in Lemma \ref{lem:convergence} and Lemma \ref{lem:convergence_forall} are exploited in order to study the interval of probabilities in which \mbox{$\lim\limits_{k\to\infty} \bold{tr}\left\{Z_{i,n}(k)\right\}=+\infty$}, and \mbox{$\lim\limits_{k\to\infty} \bold{tr}\left\{Z_{i,n}(k)\right\}<\infty$} for any initial condition \mbox{$\mathbf{Z}_0\in\HH^{Nn_x,+}$}, for \mbox{$n\in\mathbb{S}$}, \mbox{$i\in\{u,e\}$}. \\
\noindent For \mbox{$\tilde{\mathbf{Y}}=\bmat \tilde{Y}_m\emat_{m=1}^{N}\in\HH^{Nn_x,+}$}, a transition probability matrix (TPM) given by \mbox{$\tilde{P}=\left[ \tilde{p}_{mn}\right]_{m,n=1}^{N}$}, and probabilities \mbox{$\alpha_m>0,\tilde{\gamma}_m\in\left[0,1\right]$}, \mbox{$m\in\mathbb{S}$}, consider the following expression for \mbox{$k\in\NN$}, \mbox{$n\in\mathbb{S}$},
\begin{align}\label{eq:Y_n_delta}
\tilde{Y}_n (k+1) =\sum\limits_{m=1}^{N} \tilde{p}_{mn} \mathcal{X}_{\lambda}\left(\tilde{Y}_m(k),\alpha_m (k),\tilde{\gamma}_m\right),
\end{align}
\noindent with \mbox{$\lambda\in\left[0,1\right]$}, given by \mbox{$\lambda=\frac{\delta}{\sigma}$}, \mbox{$\sigma\in\left(0,1\right]$} and \mbox{$\delta\in\left[0,1\right]$}.
\\
\begin{ass}\label{ass:frac_delta_sigma}
Assume that for \mbox{$\lambda=1$}, for any initial condition \mbox{$\tilde{\mathbf{Y}}_0\in\HH^{N n_x,+}$}, there exists a positive constant depending on the initial condition \mbox{$U_{\tilde{\mathbf{Y}}_0}<\infty$}, such that 
\begin{align*}
\lim\limits_{k\to\infty} \bold{tr}\left\{\mathcal{X}_{1}\left(\tilde{Y}_m(k),\alpha_m (k),\tilde{\gamma}_m\right)\right\}\leq U_{\tilde{\mathbf{Y}}_0},
\end{align*}
\noindent\mbox{$m\in\mathbb{S}$}.
\end{ass}

\begin{lem}\label{lem:convergence}
Consider for \mbox{$k\in\NN, n\in\mathbb{S}$}, the expression of \mbox{$\tilde{Y}_n(k)$} in \eqref{eq:Y_n_delta}.\\
\noindent Under Assumptions \ref{ass:mat_A} and \ref{ass:frac_delta_sigma}, there exists a limit probability \mbox{$\delta_c\in\left[0,1\right)$} such that
\begin{itemize}
\item[\mbox{$(i)$}]\mbox{$\exists \tilde{\mathbf{Y}}_0\in\HH^{N n_x,+}$},  for \mbox{$0\leq \lambda\leq \frac{\delta_c}{\sigma}$}:\mbox{$\lim\limits_{k\to\infty} \bold{tr}\left\{\tilde{Y}_n(k)\right\}\!=\!+\infty$},
\item[\mbox{$(ii)$}] \mbox{$\forall \tilde{\mathbf{Y}}_0 \in\HH^{N n_x,+}$}, for \mbox{$\frac{\delta_c}{\sigma}<\lambda\leq 1$}, \mbox{$\exists\ 0<W_{\tilde{\mathbf{Y}}_0,\lambda}<\infty$}: \mbox{$\lim\limits_{k\to\infty}\bold{tr}\left\{\tilde{Y}_n(k)\right\}\leq W_{\tilde{\mathbf{Y}}_0,\lambda}$}, with \mbox{$W_{\tilde{\mathbf{Y}}_0,\lambda}$} positive constant depending on \mbox{$\tilde{\mathbf{Y}}_0$} and on \mbox{$\lambda$}.
\end{itemize}
\end{lem}
\begin{proof}
\noindent Consider \eqref{eq:Y_n_delta} for \mbox{$\lambda=0$}. By Assumption~\ref{ass:mat_A}, \mbox{$\rho\left(A\right)>1$}. Then, for some initial condition \mbox{$\tilde{\mathbf{Y}}_0\in\HH^{Nn_x,+}$}, for \mbox{$n\in\mathbb{S}$}, \mbox{$\lim\limits_{k\to\infty}\bold{tr}\left\{\tilde{Y}_n(k)\right\}=+\infty$}, and condition \mbox{$(i)$} is satisfied.
\\\noindent When \mbox{$\lambda=1$}, by Assumption~\ref{ass:frac_delta_sigma}, for any initial condition \mbox{$\tilde{\mathbf{Y}}_0\in\HH^{N n_x,+}$}, there exists \mbox{$0<U_{\tilde{\mathbf{Y}}_0}<\infty$}, such that 
\begin{align*}
\lim\limits_{k\to\infty} \bold{tr}\left\{\mathcal{X}_{1}\left(\tilde{Y}_m(k),\alpha_m (k),\tilde{\gamma}_m\right)\right\}\leq U_{\tilde{\mathbf{Y}}_0},
\end{align*}
 \mbox{$m\in\mathbb{S}$}. Consequently, for \mbox{$\lambda=1$} there exists a positive constant \mbox{$W_{\tilde{\mathbf{Y}}_0}$}, with \mbox{$U_{\tilde{\mathbf{Y}}_0}\leq W_{\tilde{\mathbf{Y}}_0}<\infty$}, such that
\begin{align*}
\lim\limits_{k\to\infty} \sum\limits_{m=1}^{N} \tilde{p}_{mn} \bold{tr}\left\{\mathcal{X}_{1}\left(\tilde{Y}_m(k),\alpha_m (k),\tilde{\gamma}_m\right)\right\}\leq W_{\tilde{\mathbf{Y}}_0},
\end{align*} 
\noindent and thus, from \eqref{eq:Y_n_delta}, condition \mbox{$(ii)$} is satisfied for \mbox{$\lambda=1$}.\\\noindent 
Fix a \mbox{$0<\tilde{\lambda} \leq 1$}, such that for any initial condition \mbox{$\tilde{\mathbf{Y}}_0\in\HH^{N n_x,+}$} there exists \mbox{$0<U_{\tilde{\mathbf{Y}}_0,\tilde{\lambda}} <\infty$}, for which
\begin{align}\label{eq:lim_mathcalX}
\lim_{k\to\infty} \bold{tr}\left\{\mathcal{X}_{\tilde{\lambda}}\left(\tilde{Y}_m(k),\alpha_m(k),\tilde{\gamma}_m\right)\right\}\leq U_{\tilde{\mathbf{Y}}_0,\tilde{\lambda}}
\end{align}
\noindent (see \cite[Thm. 2]{Sinopoli2004Kalman}), \mbox{$m\in\mathbb{S}$}. \\\noindent
Consequently, from \eqref{eq:lim_mathcalX}, by applying \eqref{eq:Y_n_delta}, there exists a positive constant \mbox{$W_{\tilde{\mathbf{Y}}_0,\tilde{\lambda}}$}, with \mbox{$U_{\tilde{\mathbf{Y}}_0,\tilde{\lambda}}\leq W_{\tilde{\mathbf{Y}}_0,\tilde{\lambda}}<\infty$}, such that
\begin{align*}
\lim\limits_{k\to\infty} \sum\limits_{m=1}^{N} \tilde{p}_{mn} \bold{tr}\left\{\mathcal{X}_{\tilde{\lambda}}\left(\tilde{Y}_m(k),\alpha_m (k),\tilde{\gamma}_m\right)\right\}\leq W_{\tilde{\mathbf{Y}}_0,\tilde{\lambda}}.
\end{align*} 
\noindent Pick \mbox{$\tilde{\delta}=\tilde{\lambda} \sigma$}.
Then, for \mbox{$\hat{\delta}\!>\!\tilde{\delta}$}, there exists \mbox{$\hat{\lambda}\!=\! \hat{\delta}/\sigma\!>\! \tilde{\delta}/\sigma\!=\!\tilde{\lambda}$} and by Lemma \ref{lem:mathcalX_properties}~\mbox{$(b)$}, for any initial condition \mbox{$\tilde{\mathbf{Y}}_0\in\HH^{N n_x,+}$} there exists \mbox{$W_{\tilde{\mathbf{Y}}_0,\hat{\delta}/\sigma}$}, such that
\begin{align*}
\lim\limits_{k\to\infty} \sum\limits_{m=1}^{N} \tilde{p}_{mn} \bold{tr}\left\{\mathcal{X}_{\hat{\delta}/\sigma}\left(\tilde{Y}_m(k),\alpha_m (k),\tilde{\gamma}_m\right)\right\}\leq W_{\tilde{\mathbf{Y}}_0,\hat{\delta}/\sigma},
\end{align*} 
\noindent and thus, \mbox{$(ii)$} is satisfied for \mbox{$\hat{\lambda}=\hat{\delta}/\sigma$}. Define 
\begin{align*}
&\delta_c\triangleq \inf\Big\{ \delta_* :\forall\delta>\delta_*,\forall\ \tilde{\mathbf{Y}}_0\in\HH^{N n_x,+},
\exists\ 0\!<\!W_{\tilde{\mathbf{Y}}_0,\delta/\sigma}\!<\!\infty,\\& \lim\limits_{k\to\infty}\sum\limits_{m=1}^N \tilde{p}_{mn}\bold{tr}\{\mathcal{X}_{\delta/\sigma} (\tilde{Y}_m(k),\alpha_m (k),\tilde{\gamma}_m )\}\!\leq\!W_{\tilde{\mathbf{Y}}_0,\delta/\sigma}\Big\}.
\end{align*}
\noindent (see \cite[Thm. 2]{Sinopoli2004Kalman}). Consequently, condition \mbox{$(ii)$} is satisfied.\\
\noindent The proof of the lemma is complete.
\end{proof}
\begin{lem}\label{lem:convergence_forall}
Consider for \mbox{$k\in\NN, n\in\mathbb{S}$}, the expression of \mbox{$\tilde{Y}_n(k)$} in \eqref{eq:Y_n_delta}.\\
\noindent Under Assumptions \ref{ass:mat_A} and \ref{ass:frac_delta_sigma}, there exists a limit probability \mbox{$\delta_c\in\left[0,1\right)$} such that, for any initial condition \mbox{$\tilde{\mathbf{Y}}_0 \in\HH^{Nn_x,+}$}, \mbox{$n\in\mathbb{S}$},
\begin{itemize}
\item[\mbox{$(i)$}] for \mbox{$0\leq \lambda\leq \frac{\delta_c}{\sigma}$}: \mbox{$\lim\limits_{k\to\infty} \bold{tr}\left\{\tilde{Y}_n(k)\right\}=+\infty$}, 
\item[\mbox{$(ii)$}] for \!\mbox{$\frac{\delta_c}{\sigma}<\lambda\leq 1$}, there exists a positive constant \mbox{$W_{\tilde{\mathbf{Y}}_0,\lambda}$} depending on \mbox{$\tilde{\mathbf{Y}}_0$}, \mbox{$\lambda$}, \mbox{$0<W_{\tilde{\mathbf{Y}}_0,\lambda}<\infty$} : \mbox{$\lim\limits_{k\to\infty} \bold{tr}\left\{\tilde{Y}_n(k)\right\}\leq W_{\tilde{\mathbf{Y}}_0,\lambda}$}.
\end{itemize}
\end{lem}
\begin{proof}
From Assumptions \ref{ass:mat_A} and \ref{ass:frac_delta_sigma}, by Lemma~\ref{lem:convergence}~\mbox{$(ii)$}, condition \mbox{$(ii)$} in Lemma~\ref{lem:convergence_forall} follows. \\\noindent We will prove that Lemma~\ref{lem:convergence}~\mbox{$(i)$} implies condition \mbox{$(i)$} in Lemma~\ref{lem:convergence_forall}. Let {$\tilde{\mathbf{Y}}'_0 \in\HH^{Nn_x,+}$} be one initial condition for which the sequence \mbox{$\tilde{Y}'_m(k)$} is such that \mbox{$\lim\limits_{k\to\infty} \bold{tr}\left\{\tilde{Y}'_n(k)\right\}=+\infty$} according to Lemma~\ref{lem:convergence}~\mbox{$(i)$}, for \mbox{$m,n\in\mathbb{S}$}. 
\\\noindent Also let \mbox{$\tilde{Y}_m(k)$}, \mbox{$m\in\mathbb{S}$}, \mbox{$k\in\NN$}, be the sequence with some arbitrary initial condition \mbox{$\tilde{\mathbf{Y}}_0 \in\HH^{Nn_x,+}$}. Notice that if \mbox{$\tilde{Y}_m(0)=\mathbb{O}_{n_x}$}, then \mbox{$\tilde{Y}_m(1)\succ 0$}, \mbox{$m \in\mathbb{S}$}, since \mbox{$Q\succ 0$}. Thus, we can find a large enough positive constant \mbox{$\beta>1$}, such that \mbox{$\beta \tilde{Y}_m(1)\succeq \tilde{Y}'_m(1)$}, for all \mbox{$m\in\mathbb{S}$}.\\\noindent
\begin{cl}\label{cl:Ynclaim}
\noindent There exists a positive constant \mbox{$\beta>1$}, such that \mbox{$\beta \tilde{Y}_m(k)\succeq \tilde{Y}'_m(k)$}, for all \mbox{$m\in\mathbb{S}$}, for all \mbox{$k\geq 1$}.
\end{cl}
\begin{proof}
\noindent  Assume that at time \mbox{$\overline{k}\in\NN$}, \mbox{$\beta \tilde{Y}_m(\overline{k})\succeq \tilde{Y}'_m(\overline{k})$}, \mbox{$m\in\mathbb{S}$}.\\\noindent Let us write the expression of \mbox{$\beta \tilde{Y}_n(\overline{k}+1)$}, for \mbox{$n\in\mathbb{S}$}: by the linearity of the sum and by the properties \mbox{$(c)$} and \mbox{$(a)$} in Lemma \ref{lem:mathcalX_properties}, we obtain
\begin{align*} 
\beta \tilde{Y}_n(\overline{k}+1)&\succeq \sum\limits_{m=1}^{N} \tilde{p}_{mn} \mathcal{X}_{\lambda}\left( \tilde{Y}'_m(\overline{k}),\alpha_m(\overline{k}),\tilde{\gamma}_m\right)\nonumber\\
&=\tilde{Y}'_n(\overline{k}+1),
\end{align*}
\noindent and thus, \mbox{$\beta \tilde{Y}_n(k)\succeq \tilde{Y}'_n (k)$}, for any \mbox{$k\geq 1$}, \mbox{$n\in\mathbb{S}$}.  
\end{proof}
\noindent By Claim \ref{cl:Ynclaim}, \mbox{$\lim\limits_{k\to\infty}\bold{tr} \{\tilde{Y}_m(k)\}\geq 	\frac{1}{\beta}\lim\limits_{k\to\infty}\bold{tr}\{ \tilde{Y}'_m(k)\}$}, \mbox{$n\in\mathbb{S}$} implying condition \mbox{$(i)$}.\\\noindent
The proof of the lemma is complete.
\end{proof}
\noindent The following result is exploited in the proof of Proposition \ref{prop:Lyap_eavesdropper}, used for the eavesdropper characterization.
\begin{prop}\label{prop:rho_mathcalA2}
There exists \mbox{$\mathbf{S}_e=\left[ S_{e,n}\right]_{n=1}^{N}\in\HH^{Nn_x,+}$}, satisfying \mbox{$S_{e,n}=\mathcal{S}_{e,n}\left(\mathbf{S}_e\right)$}, 
if and only if \mbox{$\rho\left(\mathcal{A}_e\right)<1$}.  
\end{prop}
\begin{proof}
By applying the vectorization to the equation \mbox{$S_{e,n}=\mathcal{S}_{e,n}\left(\mathbf{S}_e\right)$}, for \mbox{$\mathbf{S}_e=\left[ S_{e,n}\right]_{n=1}^{N}\in\HH^{Nn_x,+}$}, exploiting the properties of the Kronecker product (see for instance \cite{brewer1978kronecker}), we get the following equality,
\begin{align}\label{eq:S_n4}
\mvec^2\left(\mathbf{S}_e\right)=\mathcal{A}_e \mvec^2\left(\mathbf{S}_e\right)+\mvec^2\left(\boldsymbol\pi^{\infty}_e \otimes Q\right),
\end{align} 
\noindent with \mbox{$\boldsymbol\pi^{\infty}_e =\bmat\pi_{e,m}^{\infty}\emat_{m=1}^{N}$}. By \cite[Prop. 3.36-3.38]{COSTA2005}, \eqref{eq:S_n4} has a unique solution \mbox{$\mathbf{S}_e\in\HH^{Nn_x,+}$} if and only if \mbox{$\rho\left(\mathcal{A}_e\right)<1$}.\\\noindent The proof of the proposition is complete.
\end{proof}
\begin{proof}[Proof of Prop. \ref{prop:Lyap_eavesdropper}]
Consider for \mbox{$k\in\NN$}, \mbox{$n\in\mathbb{S}$}, the equation \mbox{$
S_{e,n}(k+1)=\mathcal{S}_{e,n}\left(\mathbf{S}_e(k)\right)$}, \mbox{$\mathbf{S}_e=\left[S_{e,n}\right]_{n=1}^N$} in \mbox{$\HH^{Nn_x,+}$}. Notice that \mbox{$S_{e,n}(k)$}, \mbox{$k\in\NN$}, \mbox{$n\in\mathbb{S}$}, is a monotonically increasing sequence, since \mbox{$Q\succ 0$}. \\\noindent Recall that \mbox{$Z_{e,n}(0)\succeq 0$} for all \mbox{$n\in\mathbb{S}$}. 
Clearly, \mbox{$\mathbb{O}_{n_x}=S_{e,n}(0)\preceq Z_{e,n}(0)$}.\\\noindent Moreover, \mbox{$S_{e,n}(k)\preceq Z_{e,n}(k)$} implies 
\begin{align*}
S_{e,n}(k+1)=\mathcal{S}_{e,n}\left(\mathbf{S}_e(k)\right)\preceq Z_{e,n}(k+1).
\end{align*}
\noindent By induction arguments, \mbox{$S_{e,n}(k)\preceq  Z_{e,n}(k)$}, for all \mbox{$k\geq 0$}.\\\noindent If \mbox{$\rho\left(\mathcal{A}_e\right)<1$}, by Prop. \ref{prop:rho_mathcalA2} \begin{align*}
\lim\limits_{k\to\infty} \bold{tr}\left\{Z_{e,n}(k)\right\}\geq\bold{tr}\left\{S_{e,n}\right\},
\end{align*}
\noindent with \mbox{$S_{e,n}=\mathcal{S}_{e,n}\left(\mathbf{S}_e\right)$}, \mbox{$\mathbf{S}_e\in\HH^{Nn_x,+}$}.\\\noindent
If \mbox{$\rho\left(\mathcal{A}_e\right)\geq 1$}, the sequence \mbox{$S_{e,n}(k)$}, \mbox{$k\in\NN$}, \mbox{$n\in\mathbb{S}$}, does not converge. Since \mbox{$S_{e,n}(k)$}, \mbox{$k\in\NN$}, \mbox{$n\in\mathbb{S}$}, is a monotonically increasing sequence, we obtain \mbox{$
\lim\limits_{k\to\infty} \bold{tr}\left\{S_{e,n}(k)\right\}=+\infty$}, implying that \mbox{$
\lim\limits_{k\to\infty} \bold{tr}\left\{Z_{e,n}(k)\right\}=+\infty$}, \mbox{$n\in\mathbb{S}$}.
\\\noindent The proof of the proposition is complete.
\end{proof}
\bibliographystyle{IEEEtran}
\bibliography{ref_2}
\end{document}